\documentclass[journal,transmag,onecolumn]{IEEEtran}
\usepackage{latexsym,amsmath,amsthm,amsfonts,amssymb,bm}
\usepackage{graphics,graphicx, epsfig,tabularx}
\theoremstyle{plain}
\newtheorem{thm}{Theorem}[section]

\newtheorem{lem}{Lemma}
\newtheorem{prop}{Proposition}
\newtheorem{defn}{Definition}
\newtheorem{rem}{Remark}

\def\F {{\mathbf{F}}}

\def\c {{\mathbf{c}}}
\def\x {{\mathbf{x}}}
\def\y {{\mathbf{y}}}
\def\v {{\mathbf{v}}}
\def\rank{{\rm rank}}

\begin{document}
%
\title{On List-decodability of Random Rank Metric Codes}


\author{\IEEEauthorblockN{Yang Ding\IEEEauthorrefmark{1,2}}
\IEEEauthorblockA{\IEEEauthorrefmark{1}Department of Mathematics, Shanghai University, P. R. China.}
\IEEEauthorblockA{\IEEEauthorrefmark{2}Division of Mathematical Science, Nanyang Technological University, Singapore.}
\thanks{This work was supported by National Natural Science Foundation of China(11201286) and a grant of "The First-Class Discipline of Universities in Shanghai". This work is also partially supported by the Singapore A*STAR SERC under Research Grant 1121720011.
Corresponding author: Y. Ding (email: dingyang@shu.edu.cn).}}

\IEEEtitleabstractindextext{%
\begin{abstract}
In the present paper, we consider list decoding for both random rank metric codes and random linear rank metric codes. Firstly, we show that, for arbitrary $0<R<1$ and $\epsilon>0$ ($\epsilon$ and $R$ are independent), if $0<\frac{n}{m}\leq \epsilon$, then with high probability a random rank metric code in $\F_{q}^{m\times n}$ of rate $R$ can be list-decoded up to a fraction $(1-R-\epsilon)$ of rank errors with constant list size $L$ satisfying $L\leq O(1/\epsilon)$. Moreover, if $\frac{n}{m}\geq\Theta_R(\epsilon)$, any rank metric code in $\F_{q}^{m\times n}$ with rate $R$ and decoding radius $\rho=1-R-\epsilon$ can not be list decoded in ${\rm poly}(n)$ time. Secondly,  we show that if  $\frac{n}{m}$ tends to a constant $b\leq 1$, then every $\F_q$-linear rank metric code in $\F_{q}^{m\times n}$ with rate $R$ and list decoding radius $\rho$ satisfies the Gilbert-Varsharmov bound, i.e., $R\leq (1-\rho)(1-b\rho)$. Furthermore, for arbitrary $\epsilon>0$ and any $0<\rho<1$, with high probability a random $\F_q$-linear rank metric codes with rate $R=(1-\rho)(1-b\rho)-\epsilon$ can be list decoded up to a fraction $\rho$ of rank errors with constant list size $L$ satisfying $L\leq O(\exp(1/\epsilon))$.
\end{abstract}

}

\maketitle

 \section{Introduction}
Rank metric codes have found various applications in  network coding, space time coding, magnetic recording and cryptography, etc. In a rank metric code, each codeword is a matrix over a finite field and the distance between two codewords is defined as the rank of their difference. The concept of rank metric was first introduced by Hua \cite{hua}, and then considered in coding theory by Delsarte \cite{dels}. By adapting the idea of Reed-Solomon code, Gabidulin \cite{Gabi10} gave a construction of a class of rank metric codes which is optimal and achieves the Singleton bound.

The problem of uniquely decoding the Gabidulin codes up to half of minimum distance has received a lot of attention in the recent years. In fact, it has been solved several times, by adapting the different
approaches for unique decoding Reed-Solomon codes to the linearized setting, starting with Gabidulin's
original paper \cite{Gabi10}, and later  in \cite{richter} a method based on the Berlekamp-Massey algorithm (or the extended Euclidean algorithm) have been proposed, and then in \cite{loid05} a method based on a Welch-Berlekamp key equation have been provided. These two methods are most efficient for high-rate and low-rate codes, respectively \cite{Gadouleau082}.

List decoding, was introduced by Elias and Wonzencraft independently, is a relax version of unique decoding for which decoder allows to output a list of possible codewords. List decoding gives a possibility  to decode beyond half of distance number of errors. A fundamental problem in list decoding is to find the tradeoff among the information rate, decoding radius and the list size. List decoding of rank metric codes has been extensively studies \cite{guruxing12,guruwang12,kkcodes,mahvar,rose,traut,wach12,wach13}. However, even for Gabidulin codes which share several common properties and results paralleling with Reed-Solomon codes, people have not found an effective list decoding algorithm with decoding radius beyond the half of the distance.
 \subsection{Known results}
 Let us briefly summarize some of previous results on list decoding of rank metric codes.\begin{itemize}
\item[(i)] For the Gabidulin codes, Loidreau \cite{loid05} gave a Welch-Berlekamp like algorithm to uniquely decode up to half of distance numbers of rank errors. However,  no polynomial-time list decoding algorithms have been found to decode beyond half of minimum distance. Moreover, we  do not know whether such an algorithm exists or not. On the other hand, Wachter-Zeh \cite{wach12} showed that for a Gabidulin code with rate $R$, if the list decoding radius is greater than Johnson radius $1-\sqrt{R}$, then list size $L$ must be exponential in the length of codes. This implies that no polynomial-time list decoding algorithms for the Gadibulin codes exist in this case (note that in the case of Reed-Solomon codes, it is still an open problem on whether  there is a polynomial-time list decoding algorithms that decodes more than  $1-\sqrt{R}$ fraction of errors.) \item[(ii)] Some variants of the Gabidulin codes can be list decoded with decoding radius up to the Singleton bound. By the Monte Carlo construction, Guruswami and Xing \cite{guruxing12} showed that there exist  subcodes of the Gabidulin codes with rate $R$ which can be efficiently list decoded up to a fraction $(1-R-\epsilon)$ of rank errors. Recently, by explicitly constructing certain subspace designs,  Guruswami and Wong \cite{guruwang13} presented a deterministic algorithm to decode the subcodes of the Gadibulin codes constructed in \cite{guruxing12}.
    Mahdavifar and Vardy \cite{mahvar} showed that one can list decode folded-Gabidulin codes with rate $R$ and decoding radius  up to $1-R-\epsilon$. However,  the output list size of their algorithm is exponential in the length of the code.
    \end{itemize}

\subsection{Our results} The current paper focuses  on list decoding of random rank metric codes.  Our contribution of this paper is two-fold. More precisely, we have the following two results.
\begin{itemize}
\item[(i)] Firstly, for any real numbers $R$ and $\epsilon$ with $0<R<1$ and $\epsilon>0$ ($\epsilon$ and $R$ are independent), and sufficiently large integers $n, m$,
\begin{itemize}\item[(1)] if $0<\frac{n}{m}\leq \epsilon$, then with high probability a random rank metric code  of rate $R$ in $\F_{q}^{m\times n}$  can be list-decoded up to a fraction $(1-R-\epsilon)$ rank errors with list size $L\leq O(1/\epsilon)$; \item[(2)] if $\frac{n}{m}\geq\Theta_R(\epsilon)$, then list size $L$ of a rank metric code  with decoding radius $1-R-\epsilon$ in $\F_{q}^{m\times n}$ must be exponential in $n$.\end{itemize}\item[(ii)]Secondly, let $b=\frac{n}{m}\leq 1$ be a constant. For $\epsilon>0$ and $0<\rho<1$, let $R=(1-\rho)(1-b\rho)-\epsilon$, we show that with high probability a random $\F_q$-linear rank metric codes  with rate $R$  is $(\rho, O(\exp(1/\epsilon)))$-list decodable. Furthermore, for any $\F_q$-linear rank metric code with list decoding radius $\rho$, its rate $R$ is at most $(1-\rho)(1-b\rho)$.
\end{itemize}

\subsection{Organiztion} The paper is organized as follows. In Section 2, we introduce notations of rank metric codes and list decoding. Moreover, a constrains on list decoding radius is derived in the same section. In Section 3 random rank metric codes are discussed. In particular, we show that almost all random rank metric codes are good list decodable codes if the ratio $n/m$ is small.  In Section 4, random linear rank metric codes are studied. Similar results are obtained by replacing the Singleton bound by the Gilbert-Varsharmov bound.  Finally, in Section 5, we draw conclusions and mention some open problems in list decoding of rank metric codes.
\section{Preliminary}
\subsection{Rank Metric Codes} Denote by $\F_q$ the finite field with $q$ elements. Let $\F_{q}^{m\times n}$ denote the set of all $m\times n$ matrices over $\F_q$.  Without loss of generality, we always assume that $n\leq m$ in this paper (indeed, if $n\geq m$, we can simply consider the transpose of matrices). For any $X,Y\in\F_{q}^{m\times n}$, the rank distance is defined by $$d_R(X,Y):=\rank(X-Y).$$
A rank metric code $C$ is just a subset of $\F_{q}^{m\times n}$. The rate and minimum rank distance of $C$ is defined by $R=\frac{\log_q|C|}{mn}$
and  $d_R(C)=\min\{\rank(X-Y):\; X\neq Y, X,Y\in C\}$, respectively. It is clear that we have $d_R(C)\leq \min\{m,n\}$.

If we fix an $\F_q$-basis of $\F_{q^m}$, each element in $\F_{q^m}$ can be identified with a column vector of $\F_q^m$, and vice versa. Thus, each matrix $X\in \F_q^{m\times n}$ can be identified with a vector $\x\in\F_{q^{m}}^{n}$. For any vectors $\x,\mathbf{y}\in\F_{q^m}^n$, the rank distance between $\x$ and $\mathbf{y}$ is defined by $d_R(\x,\mathbf{y}):=\rank(X-Y),$
where $X,Y$ are the corresponding matrices of $\x,\mathbf{y}$, respectively (note that $d_R(\x,\mathbf{y})$ is independent of the choice of an $\F_q$-basis of $\F_{q^m}$). Hence, a rank metric code $C\in\F_q^{m\times n}$ can be viewed as a block code of length $n$ over $\F_{q^m}$ with rank distance. In the remainder of the paper, a rank metric code is viewed  as a subset of $\F_{q^m}^n$ and we denote by $\rank(\x)$ to be the rank of the corresponding $m\times n$ matrix of $\x$.

Similar to  block codes, one can also derive the Singleton bound as shown below.
\begin{lem} {\rm (Singleton Bound)} \cite{Gabidulin85} Let $C\subseteq \F_{q^m}^{n}$ be a rank metric code with minimum rank distance $d$, then
$$\log_q|C|\leq \min\{n(m-d+1),m(n-d+1)\}.$$
Hence, $\log_q|C|\le m(n-d+1)$ since $n\le m$.\end{lem}
A rank metric code $C$  achieving the above Singleton  bound is called maximum rank distance (MRD) code. For example, Gabidulin codes are a class of MRD codes.

One can also define rank metric balls as follows.
\begin{defn} (Rank Metric Ball) For a vector $\x\in\F_{q^m}^n$ and a nonnegative real number $r$, the rank metric ball of center $\x$ with radius $r$ is define by $$B_R(\x,r):=\{\mathbf{y}\in\F_{q^m}^n:\; d_R(\x,\mathbf{y})\leq r\}.$$\end{defn}
It is clear  that the volume of a rank metric ball is independent of the choice of the center $\x$ and depends only on its radius.

Let $N_u(q^m,n)$ denote the number of vectors in $\F_{q^m}^n$ with rank $u$. It is shown in \cite{Gabidulin85}  that \[N_u(q^m,n)=\left\{\begin{array}{ll}
1& \mbox{if $u=0$}\\
\prod\limits_{i=0}^{u-1}\frac{(q^n-q^i)(q^m-q^i)}{q^u-q^i} &\mbox{if $u\geq 1$.}\end{array}\right.\]
Hence, $$|B_R(0,r)|=1+\sum_{u=1}^{r}\prod_{i=0}^{u-1}\frac{(q^n-q^i)(q^m-q^i)}{q^u-q^i}.$$
Furthermore, the volume of a rank metric ball can be bounded as shown below.
 \begin{lem}\label{gad}\cite{Gadouleau08} For $0\leq r\leq \min\{m,n\}$, one has $$q^{r(m+n-r)}< |B_R(0,r)|< K_q^{-1}q^{r(m+n-r)},$$where $K_q=\prod_{j=1}^{\infty}(1-q^{-j})$.\end{lem}

It is easy to see that $ K_q$ is an increasing function of $q$. In  \cite{Helleseth95}, it is shown that $K_2\approx 0.2887$. Hence, we have $K_q^{-1}\leq K_2^{-1}<4$.

Since rank metric defines a distance, we have the Hamming and Gilbert-Varsharmov like bounds as well.\

\begin{lem} {\rm (Hamming Bound)} \cite{Gadouleau06}  Let $C\subseteq \F_{q^m}^{n}$ be a rank metric code with minimum rank distance $d$, then
$$\left|B_R\left(0,\frac{d-1}{2}\right)\right|\leq \frac{q^{mn}}{|C|}.$$\end{lem}

A rank metric code $C$  achieving the above Hamming  bound is called a perfect code. However, unlike classical codes, it is shown in \cite{babu95} that  there exist no perfect rank metric codes.

There are both finite and asymptotic versions of the covering bound (i.e., the Gilbert-Varsharmov  bound) for rank metric codes. In this paper, we only state the  asymptotic versions of the Gilbert-Varsharmov  bound. The reader  may refer to  \cite{Gadouleau08} for the finite version of the Gilbert-Varsharmov  bound.

Define $A_{q^m}(n,d)$ to be the maximum cardinality of rank metric codes with rank distance $d$ in $\F_{q^m}^{n}$, i.e.,
$$A_{q^m}(n,d):=\max\{|C|\;:\;C\subseteq \F_{q^m}^{n}\;\mbox{ is a rank metric code }\;\mbox{with}\;d_R(C)=d\}.$$

Furthermore, assume that $n/m$ tends to a positive constant $b$ as $n\rightarrow\infty$. For a real $\delta\in(0,1)$, define
$$ \alpha(\delta):=\limsup_{n\rightarrow\infty}\frac{\log_qA_{q^m}(n,\lfloor\delta n\rfloor)}{mn}.$$
Then we have the following asymptotic Gilbert-Varsharmov  bound.
\begin{lem} {\rm (Gilbert-Varsharmov  Bound)} \cite{Gadouleau08}  For all real numbers  $b$ and $\delta$ satisfying $0\leq\delta\leq \min\{1, b^{-1}\}$, one has $$\alpha(\delta)\geq (1-\delta)(1-b\delta).$$\end{lem}
\begin{rem}{\rm \begin{itemize}\item[(i)] Again, unlike in classical coding theory, the  quality $A_{q^m}(n,d)$ and the function $\alpha(\delta)$ are completely determined in the case of rank metric codes, i.e., $\log_qA_{q^m}(n,d)=m(n-d+1)$ and $\alpha(\delta)=1-\delta$.
\item[(ii)] In view of the above result in (i), the Gilbert-Varsharmov  bound does not make sense in determining the function $\alpha(\delta)$ . However, it does make sense in our list decoding of linear random rank metric codes in Section 4.
\end{itemize}
}\end{rem}

\subsection{List Decoding}
In this subsection, we introduce some notations for list decoding of rank metric codes and show a result on constrains of list decoding radius.
\begin{defn} ($(\rho, L)$-list decodable) A rank metric code $C\subseteq\F_{q^m}^{n}$ is said to be $(\rho, L)$-list decodable if for every $\x\in\F_{q^m}^{n}$, we have $$|B_R(\x,\rho n)\cap C|\leq L.$$\end{defn}
\begin{defn} (Decoding Radius) For an integer $L\geq 1$ and a rank metric code $C\subseteq\F_{q^m}^{n}$, the normalized list-of-L decoding radius of $C$, denoted $\rho_L(C)$, is defined by
$$\rho_L(C):=\max\{s/n|\;C\; \mbox{is}\; (s/n, L)-\mbox{list decodable}, s\in\mathbb{Z}_{\geq 0} \}.$$
\end{defn}
The following proposition shows an upper bound on decoding radius of rank metric codes.

\begin{prop} Let $m,n, L$ be positive integers satisfying   $L=O({\rm poly}(mn))$. Then for any $R\in(0,1)$ and $\rho\in (0,1)$, a $(\rho,L)$-list decodable code $C\subseteq\F_{q^m}^n$  with rate $R$ must obey $$\rho\leq 1-R.$$\end{prop}
\begin{proof}
 Let $C$ be a rank metric code in $\F_{q^m}^n$ with rate $R$.    We claim that if $\rho>1-R$, then there exists a vector $\x\in\F_{q^m}^n$ such that
 the list set $ B_R(\x,\rho n)\cap C $ of $\x$ has size at least $\Omega(\exp(nm))$.

For a codeword $\c\in C$, if $\y\in B_R(\c,\rho n)$, then $\c\in(B_R(\y, \rho n)\cap C)$. Thus
$$\sum_{\y\in \F_{q^m}^{n}}|B_R(\y,\rho n)\cap C|\geq |C||B_R(0,\rho n)|\geq q^{Rmn+\rho n(m+n-\rho n)},$$
where the last inequality follows from Lemma \ref{gad}.

By the Pigeonhole Principle, there exists $\x\in \F_{q^m}^{n}$ such that
\begin{equation}\label{3}|B_R(\x,\rho n)\cap C|\geq \frac{q^{Rmn+\rho n(m+n-\rho n)}}{q^{mn}}\geq q^{nm(R+\rho-1)}.\end{equation}
 Thus, we have $|B_R(\x,\rho n)\cap C|=\Omega(\exp(nm))$ if $\rho>1-R$. This completes the proof.\end{proof}
\section{Random Rank Metric Codes}
This section is devoted to the list decodabilty of random rank metric codes. More precisely speaking, we show that for $n/m\le \epsilon$, almost all rank metric codes of rate $R$ in $\F_{q^m}^{ n}$ are  $(1-R-\epsilon, O(1/\epsilon))$-list decodable; while for $n/m\geq \Theta_R(\epsilon)$, no rank metric codes of rate $R$ and radius $1-R-\epsilon$ in $\F_{q^m}^{ n}$ have polynomial list size.

\begin{thm}\label{thm1}For every $\epsilon > 0$ and  $0< R< 1$, with high probability a random rank metric code in $\F_{q^m}^n$ with rate $R$  is $(1-R-\epsilon, O(1/\epsilon))$-list decodable for all sufficiently large $m,n$ with $n/m\le \epsilon$.  \end{thm}

\begin{proof}  Pick a rank metric code $C\subseteq \F_{q^m}^n$  with size $M=|C|=q^{Rmn}$   uniformly at random. Letting $L+1=\frac{4}{\epsilon}, \rho=1-R-\epsilon$. Next we calculate the probability that $C$ is not $(\rho,L)$-list decodable.

If $C$ is not $(\rho,L)$-list decodable, then there exists $\x\in\F_{q^m}^n$, and a subset $S\subseteq C$ with $|S|=L+1$ such that $S\subseteq B_R(\x,\rho n)$.
For a fixed $\x$, the probability that one codeword of $C$ is contained in  $B_R(\x,\rho n)$ is at most $\frac{|B_R(\x,\rho n)|}{q^{mn}}$. For a subset $S\subseteq C$ with $|S|=L+1$, the probability of $S\subseteq B_R(\x,\rho n)$ is at most $\left(\frac{|B_R(\x,\rho n)|}{q^{mn}}\right)^{L+1}$. Furthermore, the number of subsets $S\subseteq C$ with $|S|=L+1$ is at most $M^{L+1}$. By the union bound, $C$ is not $(\rho, L)$-list decodable with probability at most
\begin{eqnarray*}q^{mn}M^{L+1}\left(\frac{|B_R(\x,\rho n)|}{q^{mn}}\right)^{L+1}&=&q^{mn}q^{Rmn(L+1)}K_q^{-(L+1)}q^{(L+1)[\rho n(m+n-\rho n)-mn]} \\
&< & q^{mn}q^{(L+1)[-\epsilon mn+\rho(1-\rho)n^2+\log_q 4]}\\
&\leq & q^{mn}q^{(L+1)(-\epsilon mn+\frac{n^2}{2})}\quad( n\;\mbox{sufficiently large})\\
&\leq &q^{mn}q^{(L+1)(-\frac{\epsilon}{2} mn)}\quad(\mbox{since}\; n\leq \epsilon m)\\
&\leq & q^{-mn}. \end{eqnarray*}
This completes the proof.\end{proof}
\begin{rem}\label{rem1}\begin{itemize}\item[(1)]Using some probabilistic arguments, we can easily show the following: a random rank metric code with rate $R$ is $(1-R-\epsilon, L)$-list decodable, then with high probability the list size $L\geq 4/\epsilon$. In other words, for a random $(1-R-\epsilon, L)$-list decodable rank metric code of rate $R$, when $\epsilon\rightarrow 0$, the list size $L$ tends to $\infty$ with high probability.
\item[(2)]Furthermore, the ratio of $n/m=\epsilon$ in Theorem \ref{thm1} is optimal in magnitude. The next theorem shows that there exist no rank metric codes of rate $R$ and radius $1-R-\epsilon$ in $\F_{q^m}^{ n}$ with polynomial list size for $n/m\ge \Theta_R(\epsilon)$, i.e., no polynomial-time list decoding algorithm is possible in this case.\end{itemize}\end{rem}

\begin{thm}For every $\epsilon >0$ and $0<R<1$. Let $C\in\F_{q^m}^n$ be a rank metric code with rate $R$ and list decoding radius $1-R-\epsilon$, then the maximal list set size of $C$ is at least $\Omega({\rm exp}(n))$ for all sufficiently large $n,m$ with $ \frac{n}{m}\geq \Theta_R(\epsilon)$, where $\Theta_R(\epsilon)=\frac{2\epsilon}{(1-R-\epsilon)(R+\epsilon)}$.
\end{thm}
\begin{proof}Fix $\epsilon\in(0, 1)$, let $m\geq 1/\epsilon$ be an integer. Let $C\subseteq\F_{q^m}^n$ be a random rank metric code with rate $R$.  Put $\rho=1-R-\epsilon$, we claim that if $n/m\geq\frac{2\epsilon}{(1-R-\epsilon)(R+\epsilon)}$, then there exists a vector $\x\in\F_{q^m}^n$ such that $|B_R(\x,\rho n)\cap C |$ is at least $\Omega({\rm exp}(n))$.

For any codeword $\c\in C$, if $\y\in B_R(\c,\rho n)$, then $\c\in B_R(\y, \rho n)\cap C$. Thus
$$\sum_{\y\in \F_{q^m}^{n}}|B_R(\y,\rho n)\cap C|\geq |C||B_R(0,\rho n)|\geq q^{Rmn+\rho n(m+n-\rho n)},$$
where the last inequality follows from Lemma \ref{gad}.

By the Pigeonhole principle, there exists $\x\in \F_{q^m}^{n}$, such that
\begin{equation}\label{1}|B_R(\x,\rho n)\cap C|\geq \frac{q^{Rmn+\rho n(m+n-\rho n)}}{q^{mn}}=q^{n[-\epsilon m+(1-R-\epsilon)(R+\epsilon)n]}.\end{equation}
Since $n\geq \frac{2\epsilon}{(1-R-\epsilon)(R+\epsilon)}m$ and $m>1/\epsilon$, we have $-\epsilon m+(1-R-\epsilon)(R+\epsilon)n\geq 1$, then (\ref{1}) shows that $|B_R(\x,\rho n)\cap C |$ is at least $\Omega({\rm exp}(n))$. This completes the proof.\end{proof}

\section{Random $\F_q$-linear Rank metric codes}
Inspired by the work of Guruswami et.al.\cite{guru10}, we consider the list decodability of random $\F_q$-linear rank metric codes where the code is viewed as an $\F_q$-linear subspace of $\F_{q^{m}}^n$. We show that if $n/m$ tends to a constant $b$ as $n\rightarrow\infty$, almost all random $\F_q$-linear rank metric code with rate $R$ and decoding radius $\rho$ satisfy the asymptotic Gilbert-Varsharmov bound. Now, we introduce the notations of $\F_q$-linear rank metric codes as follows.

\begin{defn}An $\F_q$-linear rank metric code $C\in\F_{q^m}^n$ is an $\F_q$-linear subspace of $\F_{q^{m}}^n$. The rate of $C$ is $R=\frac{\log_q|C|}{mn}$.\end{defn}

Let $\lim_{n\rightarrow\infty}\frac{n}{m}=b$ be a constant. Now we want to show that for random $\F_q$-linear rank metric codes, with high probability, the decoding radius $\rho$ and the rate $R$ achieving the Gilbert-Varsharmov bound. Our main result is as follows.

\begin{thm}\label{thm3}Let $b=\lim_{n\rightarrow \infty}\frac{n}{m}$ be a constant. For any $\epsilon>0$ and $0< \rho< 1$, let $R=(1-\rho)(1-b\rho)-\epsilon$, then with high probability a random $\F_q$-linear rank metric code $C\subseteq\F_{q^m}^n$ of rate $R$ is $(\rho,O({\rm exp}(1/\epsilon)))$-list decodable for all sufficiently large $n, m$.\end{thm}
\begin{proof} Fix $\epsilon\in(0, 1)$, let $m=\lceil4/\sqrt{b\epsilon}\rceil, n=bm$. Let $C$ be a random $\F_q$-linear subspace of $\F_{q^m}^n$ with $\dim_{\F_q}C=Rnm$. Put $l=\lceil4/\epsilon\rceil, L=q^l$. Next we calculate the probability that $C$ is not $(\rho, L)$-list decodable.

If $C$ is not $(\rho, L)$-list decodable, then there exists $\x\in\F_{q^m}^n$ such that $|B_R(\x,\rho n)\cap C|\geq L$. We claim that $C$ is not $(\rho, L)$-list decodable with probability at most $q^{-nm/2}$, i.e.,
\begin{equation}\label{5.1}Pr[\exists \x\in\F_{q^m}^n\;s.t. \;|B_R(\x,\rho n)\cap C|\geq L]< q^{-nm/2}.\end{equation}

Let $\x\in\F_{q^m}^n$ be picked uniformly at random, define
$$\Delta:=Pr_{\x}[|B_R(\x,\rho n)\cap C|\geq L].$$
To prove inequality (\ref{5.1}), it suffices to show that $$\Delta< q^{-nm/2}q^{-mn}.$$

Since $C$ is $\F_q$-linear, we have
\begin{eqnarray*}\Delta&=&Pr_{\x}[|B_R(\x,\rho n)\cap C|\geq L]\\
&=&Pr_{\x}[|B_R(0,\rho n)\cap (C+\x)|\geq L]\\
&\leq&Pr_{\x}[|B_R(0,\rho n)\cap Span_{\F_q}(C,\x)|\geq L]\\
&\leq&Pr_{C^*}[|B_R(0,\rho n)\cap C^*|\geq L],\end{eqnarray*}
where $C^*$ is a random $\F_q$-subspace of $\F_{q^m}^n$ with dimension $Rnm+1$ containing $C$ (If $\x\notin C$, then $C^*=Span_{\F_q}(C,\x)$; otherwise $C^*=Span_{\F_q}(C, \y)$, where $\y$ is picked randomly from $\F_{q^m}^n\setminus C$).

For each integer $r$ with $l\leq r\leq L$, let $\mathcal{F}_r\subseteq B_R(0,\rho n)^r$ be the set of all $r$-tuple $(\v_1,\v_2,..,\v_r)$ such that $\v_1,\v_2,..,\v_r$ are $\F_q$-linearly independent and $$|{\rm Span}_{\F_q}(\v_1,\v_2,..,\v_r)\cap C|\geq L.$$
Obviously, $|\mathcal{F}_r|\leq |B_R(0,\rho n)|^r\leq \left(4q^{mn(\rho+\rho b-\rho^2 b)}\right)^r.$

Let $\mathcal{F}=\bigcup_{r=l}^L\mathcal{F}_r.$ If $|B_R(0,\rho n)\cap C^*|\geq L$, then there must exist an integer $r, l\leq r\leq L$ and $\v=(\v_1,...,\v_r)\in \mathcal{F}$ such that $\{\v_1,...,\v_r\}\subseteq C^*$. For simplicity, we can take $\{\v_1,...,\v_r\}$ to be a maximal $\F_q$-linearly independent subset of $B_R(0,\rho n)\cap C^*$. Thus we have,

\begin{eqnarray*}\Delta &\leq& Pr_{C^*}[|B_R(0,\rho n)\cap C^*|\geq L]\leq \sum_{\v\in\mathcal{F}}Pr_{C^*}[\{\v_1,..,\v_r\}\subseteq C^*]\\
&=&\sum_{r=l}^L\sum_{\v\in\mathcal{F}_l}Pr_{C^*}(\{\v_1,..,\v_r\}\subseteq C^*)\end{eqnarray*}

Let $\v=(\v_1,..,\v_r)\in\mathcal{F}_r$, then $\v_1,..,\v_r$ are linearly independent over $\F_q$, Therefore, we have $$Pr_{C^*}[\{\v_1,..,\v_r\}\subseteq C^*]=\prod_{i=0}^{r-1}\frac{q^{Rmn+1}-q^i}{q^{mn}-q^i}\leq \left(\frac{q^{Rnm+1}}{q^{nm}}\right)^r.$$

Thus
\begin{eqnarray*}\Delta&\leq & \sum_{r=l}^L\sum_{\v\in\mathcal{F}_r}\left(\frac{q^{Rnm+1}}{q^{nm}}\right)^r=\sum_{r=l}^L|\mathcal{F}_r|\left(\frac{q^{Rnm+1}}{q^{nm}}\right)^r\\
&\leq &\sum_{r=\lceil4/\epsilon\rceil}^{L}\left(\frac{q^{Rnm+1}}{q^{nm}}\right)^rq^{mnr(\rho+\rho b-\rho^2 b)+r\log_q4}\\
&\leq&\sum_{r=\lceil4/\epsilon\rceil}^{L}q^{mnr(\rho+\rho b-\rho^2 b+R-1+\frac{1+\log_q4}{mn})}\\
&\leq& Lq^{-2mn}\quad(\mbox{replace by}\; R=(1-\rho)(1-b\rho)-\epsilon, m=\lceil\frac{4}{\sqrt{b\epsilon}}\rceil, n=bm)\\
&\leq&q^{-mn/2}q^{-mn}\quad( \mbox{for sufficiently large}\;m, n ).\end{eqnarray*}
This completes the proof.\end{proof}
\begin{rem}When $b=0$, $R=1-\rho-\epsilon$ and $R=(1-\rho)(1-b\rho)-\epsilon$ coincide. This means that when $m$ increases faster than $n$, the decoding radius and the rate of random $\F_q$-linear rank metric codes satisfy the Singleton bound with high probability, and the list size is $O(\exp(1/\epsilon))$ (note that in this case, the list size is $O(1/\epsilon)$ for random rank metric codes).\end{rem}

The next theorem gives a constrains of list decoding radius for $\F_q$-linear rank metric codes. More precisely speaking, we prove that for any $\F_q$-linear rank metric code with list decoding radius $\rho$, its rate $R$ is bounded by the Gilbert-Varsharmov bound.
\begin{thm}Let $b=\lim_{n\rightarrow \infty}\frac{n}{m}$ be a constant. Then for any $R\in(0,1)$ and $\rho\in (0,1)$, a $(\rho, L)$-list decodable $\F_q$-linear rank metric code $C\subseteq \F_{q^m}^n$ with rate $R$ must satisfy
$$R\leq (1-\rho)(1-b\rho).$$\end{thm}
\begin{proof}Let $C\subseteq\F_{q^m}^n$ be an $\F_q$-linear rank metric code with rate $R$. We claim that if $R> (1-\rho)(1-b\rho)$, then there exists a vector $\x\in\F_{q^m}^n$ such that $|B_R(\x,\rho n)\cap C |$ is exponential in $nm$.

For any $\y\in\F_{q^m}^n$, we have
$$|B_R(\y,\rho n)\cap C|=|B_R(0,\rho n)\cap (C+\y)|.$$
Since $C$ is $\F_q$-linear, $C$ gives a total disjoint partition of $\F_{q^m}^n$. Let $S$ denoted the set of coset representatives of $C$ in $\F_{q^m}^n$. Then $|S|=\frac{q^{mn}}{|C|}=q^{(1-R)mn}$, and we have
$$B_R(0, \rho n)=\bigcup_{\y\in S}^{\bullet}(B_R(0, \rho n)\cap (C+\y)).$$
i.e.,
$$\sum_{\y\in S}|(B_R(0, \rho n)\cap (C+\y))|=|B_R(0, \rho n)|$$

By the Pigeonhole principle, there exists $\x\in \F_{q^m}^{n}$ such that
\begin{equation}\label{3}|B_R(0,\rho n)\cap (C+\x)|\geq \frac{|B_R(0, \rho n)|}{|S|}\geq \frac{q^{mn(\rho+b\rho-b\rho^2)}}{q^{(1-R)mn}}=q^{mn(R-(1-\rho)(1-b\rho))}.\end{equation}
By the given condition in the theorem, we have $-(1-\rho)(1-b\rho)+R>0$, then (\ref{3}) implies that list size of $\x$ is exponential in $nm$. This completes the proof.\end{proof}
\begin{rem}The result in \cite{wach12} shows that if the decoding radius of Gabidulin code is greater than the Johnson bound $1-\sqrt{R}$, the list size must be exponential. This result can be viewed as a special case of our theorem for $b=1$.\end{rem}

\section{Conclusion}
In this paper, we discuss the list decodability of rank metric codes. Our results show that: (1) For rank metric codes, the Singleton bound $R=1-\rho$ is the list decoding barrier. Moreover, with high probability, the decoding radius and the rate of random rank metric codes satisfy the Singleton bound for small ratio $n/m$. (2) For $\F_q$-linear rank metric codes with $n/m$ tends to constant $b$, the Gilbert-Varsharmov bound $R=(1-\rho)(1-b\rho)$ is the list decoding barrier. Furthermore, with high probability, the decoding radius and the rate of random $\F_q$-linear rank metric codes satisfy the Gilbert-Varsharmov bound. Figure 1 show these results (where $b=1/2$).

To end this paper, we propose two open problems which might be worth of investigating:
\begin{itemize}\item[(1)]What is the list decoding barrier for linear rank metric codes (i.e., $\F_{q^m}$-subspace of $\F_{q^m}^n$)?\item[(2)] For a random code $C\in\F_q^n$ with rate $R$, with high probability $C$ is $(1-R-\epsilon, L)$-list decodable, where $\Omega(\log(1/\epsilon))\leq L\leq O(1/\epsilon)$. For rank metric codes, Theorem \ref{thm1} shows that almost all random rank metric codes are $(1-R-\epsilon, L)$-list decodable with $L\leq O(1/\epsilon)$. Moreover, Remark \ref{rem1} shows that with high probability the list size $L$ tends to $\infty$ when $\epsilon\rightarrow 0$. Can we derive a reasonable lower bound of the list size?\end{itemize}

\begin{center}\begin{figure}
  \includegraphics[width=7.5in]{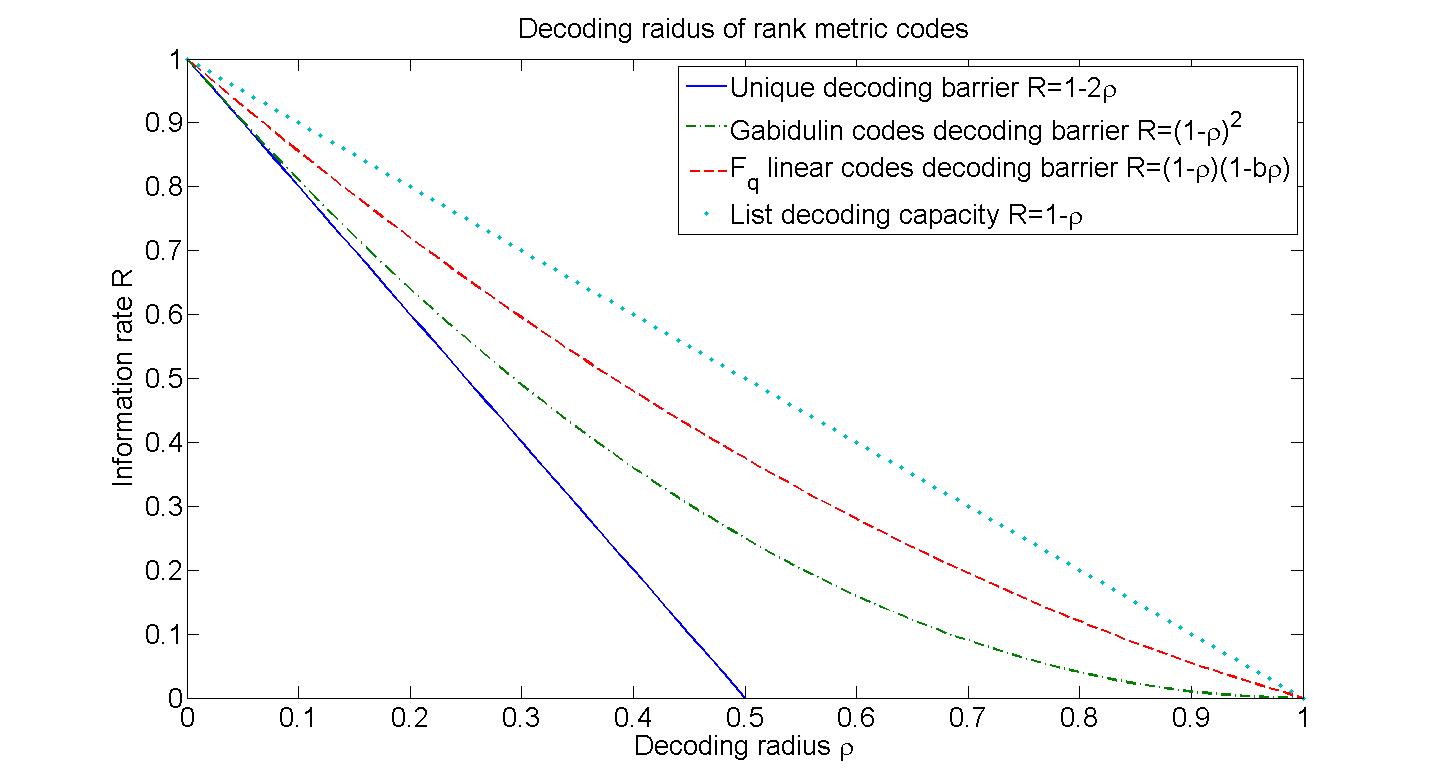}\\
  \caption{Decoding radius of rank metric codes}\label{fig1}
\end{figure}\end{center}
\section*{Acknowledgment}

The author is grateful to Prof. Xing Chaoping for his guidance and several useful discussions.

%
%
%




\end{document}